\documentclass[11pt,letterpaper]{article}

\usepackage[margin=1in]{geometry}

\usepackage{theorem,latexsym,graphicx,amssymb}
\usepackage{amsmath,enumerate}
\usepackage{wrapfig}
\usepackage{subfigure}
\usepackage{float}
\usepackage{psfig}
\usepackage{epsfig}
\usepackage{xspace}
\usepackage{paralist}
\usepackage[compact]{titlesec}
\usepackage[boxed,lined]{algorithm2e}
\usepackage{enumerate}
\usepackage{cases}
\usepackage{caption}

\usepackage{color}
\definecolor{Darkblue}{rgb}{0,0,0.4}
\definecolor{Brown}{cmyk}{0,0.81,1.,0.60}
\definecolor{Purple}{cmyk}{0.45,0.86,0,0}
\usepackage[breaklinks]{hyperref}
\hypersetup{colorlinks=true,
            citebordercolor={.6 .6 .6},linkbordercolor={.6 .6 .6},%
citecolor=black,urlcolor=black,linkcolor=black,pagecolor=black}
\newcommand{\lref}[2][]{\hyperref[#2]{#1~\ref*{#2}}}


\usepackage{multicol}

\newenvironment{proof}{{\bf Proof:  }}{\hfill\rule{2mm}{2mm}}
\newenvironment{proofof}[1]{{\bf Proof of #1:  }}{\hfill\rule{2mm}{2mm}}

\numberwithin{figure}{section}
\numberwithin{equation}{section}

\newtheorem{definition}{Definition}[section]

\newtheorem{theorem}{Theorem}[section]
\newtheorem{lemma}{Lemma}[section]

\newtheorem{fact}{Fact}[section]
\newtheorem{observation}{Observation}[section]

\providecommand{\Appendix}{}
\renewcommand{\Appendix}[2][?]{%
        \refstepcounter{section}%
        \vspace{\parskip}%
        {\flushright\large\bfseries\appendixname\ \thesection: #1}%
        \vspace{\baselineskip}%
}

\renewcommand{\appendix}{%
        \newpage
        \renewcommand{\section}{\secdef\Appendix\Appendix}%
        \renewcommand{\thesection}{\Alph{section}}%
        \setcounter{section}{0}%
}

\newcommand{\Z}{\mathbb{Z}}

\newcommand{\txts}{\textstyle}
\newcommand{\LPn}{\ensuremath{\mathsf{LP}_n}}
\newcommand{\LPi}{\ensuremath{\mathsf{LP}_\infty}}
\newcommand{\LDi}{\ensuremath{\mathsf{LD}_\infty}}
\newcommand{\CP}{\ensuremath{\mathsf{CP}}}
\newcommand{\CD}{\ensuremath{\mathsf{CD}}}

\newcommand{\mrg}{\textsf{MRG}\xspace}
\newcommand{\ranking}{\textsf{Ranking}\xspace}
\newcommand{\gmg}{\textsf{Greedy Matching Game}\xspace}
\newcommand{\obm}{\textsf{Online Bipartite Matching}\xspace}
\newcommand{\U}{\ensuremath{\mathcal{U}}}

\newcounter{note}[section]

\newcommand{\initOneLiners}{%
    \setlength{\itemsep}{0pt}
    \setlength{\parsep }{0pt}
    \setlength{\topsep }{0pt}
}

\newcommand{\ignore}[1]{}

\newcommand*\samethanks[1][\value{footnote}]{\footnotemark[#1]}

\newcommand{\shortv}[1]{}

\title{Ranking on Arbitrary Graphs: Rematch via Continuous LP with Monotone and Boundary Condition Constraints}

\author{T-H. Hubert Chan\thanks{Department of Computer Science, the University of Hong Kong. {\texttt{\{hubert,fchen,xwwu,zczhao\}@cs.hku.hk}}} \and Fei Chen\samethanks \and Xiaowei Wu\samethanks \and Zhichao Zhao\samethanks}
\date{}

\begin{document}

\begin{titlepage}

\maketitle

\begin{abstract}
Motivated by online advertisement and exchange settings, greedy randomized algorithms for the maximum matching problem have been studied, in which the algorithm makes (random) decisions that are essentially oblivious to the input graph.  Any greedy algorithm can achieve performance ratio 0.5, which is the expected number of matched nodes to the number of nodes in a maximum matching.

Since Aronson, Dyer, Frieze and Suen proved that the Modified Randomized Greedy (MRG) algorithm achieves performance ratio $0.5 + \epsilon$ (where $\epsilon = \frac{1}{400000}$) on arbitrary graphs in the mid-nineties, no further attempts in the literature have been made to improve this theoretical ratio for arbitrary graphs until two papers were published in FOCS 2012.
Poloczek and Szegedy also analyzed the MRG algorithm to give ratio 0.5039, while Goel and Tripathi used experimental techniques to analyze the Ranking algorithm to give ratio 0.56.  However, we could not reproduce the experimental results of Goel and Tripathi.

In this paper, we revisit the Ranking algorithm using the LP framework. Special care is given to analyze the structural properties of the Ranking algorithm in order to derive the LP constraints, of which one known as
the \emph{boundary} constraint requires totally new analysis and is crucial to the success of our LP.

We use continuous LP relaxation to analyze the limiting behavior as the finite LP grows. Of particular interest are new duality and complementary slackness characterizations that can handle the monotone and the boundary constraints in continuous LP.  We believe our work achieves the currently best theoretical performance ratio of $\frac{2(5-\sqrt{7})}{9} \approx 0.523$ on arbitrary graphs. Moreover, experiments suggest that Ranking cannot perform better than $0.724$ in general.
\end{abstract}

\thispagestyle{empty}
\end{titlepage}

\section{Introduction} \label{sec:intro}

Maximum matching~\cite{MicaliV80} in undirected graphs
is a classical problem in computer science.
However, as motivated by online advertising~\cite{Goel2008, Aggarwal2011} and exchange settings~\cite{Roth04},
information about the graphs can be incomplete or unknown.
Different online or greedy versions of the problem~\cite{Aronson1995,Matthias12,Goel12} can be formulated
by the following game, in which the algorithm is
essentially oblivious to the input graph.

\noindent \gmg.
An \emph{adversary} commits to a graph $G(V,E)$ and
reveals the nodes $V$ (where $n = |V|$) to the (possibly randomized) \emph{algorithm}, while keeping
the edges $E$ secret.  The algorithm returns a list $L$ that gives
a permutation of the set ${V \choose 2}$ of unordered pairs of nodes.
Each pair of nodes in $G$ is probed according to the order specified
by $L$ to form a matching greedily.
In the round when
a pair $e = \{u,v\}$ is probed,
if both nodes are currently \emph{unmatched} and
the edge $e$ is in $E$, then the two nodes will be \emph{matched}
to each other; otherwise, we skip to the next pair in $L$ until all pairs in $L$ are probed.
The goal 
is to maximize the \emph{performance ratio} of the (expected) number of nodes matched by the algorithm to the number of nodes
in a maximum matching in $G$.

Observe that any ordering of the pairs ${V \choose 2}$ will result in
a maximal matching in $G(V,E)$, giving a trivial performance ratio at least $0.5$.  However, for any deterministic algorithm,
the adversary can choose
a graph such that ratio $0.5$ is attained.  
The interesting question is: how much
better can randomized algorithms perform on arbitrary graphs?
(For bipartite graphs, there are theoretical analysis of randomized algorithms~\cite{Karande2011,Mahdian2011} achieving 
ratios better than 0.5.)

The \ranking algorithm (an early version appears in \cite{Karp90})
 is simple to describe: a permutation $\sigma$
on $V$ is selected uniformly at random, and naturally induces a lexicographical
order on the unordered pairs in ${V \choose 2}$ used for probing.
Although by experiments, the \ranking algorithm and other randomized algorithms seem to achieve performance ratios much larger than $0.5$, until very recently,
the best theoretical performance ratio $0.5 + \epsilon$ (where $\epsilon = \frac{1}{400000}$) 
on arbitrary graphs
was proved in the mid-nineties by Aronson et al.~\cite{Aronson1995}, who analyzed
the \textsf{Modified Randomized Greedy} algorithm (\mrg), which can be viewed as
a modified version of the \ranking algorithm.

After more than a decade of research, two papers were published in FOCS 2012 that
attempted to give theoretical ratios significantly better than the $0.5 + \epsilon$ bound.
Poloczek and Szegedy~\cite{Matthias12} also analyzed the \mrg algorithm 
to give ratio $0.5+\frac{1}{256} \approx 0.5039$, and 
Goel and Tripathi~\cite{Goel12} analyzed the \ranking algorithm
to give ratio $0.56$; however, we could not reproduce the experimental
results in~\cite{Goel12}.  Both papers used a common framework which has 
been successful for analyzing bipartite graphs: (i) utilize
the structural properties of the matching problem to form a minimization
linear program that gives a lower bound on the performance ratio; (ii) analyze
the LP theoretically and/or experimentally to give a lower bound.

\ignore{
We retrace the footsteps of previous researchers and discover that the case
for arbitrary graphs is harder than it seems. For step (i), upon close scrutiny,
we could not complete the arguments presented in both papers~\cite{Matthias12,Goel12}
that are needed to produce the constraints in the corresponding LPs; for step (ii),
we could not reproduce the experimental and theoretical analysis given in \cite{Goel12} that
produces the ratio 0.56.
}

In this paper, we revisit the \ranking algorithm using
the same framework:  (i) we use novel techniques to carefully
analyze the structural properties of \ranking
for producing new LP constraints; (ii) moreover,
we develop new primal-dual techniques for continuous LP to analyze 
the limiting behavior as the finite LP grows. Of particular interest
are new duality and complementary slackness results that can handle
monotone constraints and boundary conditions in continuous LP.
We believe that this paper achieves the currently best theoretical 
performance ratio of  $\frac{2(5-\sqrt{7})}{9} \approx 0.523$ on arbitrary graphs.  As a side note, our experiments
suggest that \ranking cannot perform better than $0.724$ in general.

\subsection{Our Contribution and Techniques}  


\begin{theorem} \label{th:ratio}
For the \gmg on arbitrary graphs, the \ranking algorithm achieves performance ratio
at least $\frac{2(5-\sqrt{7})}{9} \approx 0.523$.
\end{theorem}

Following previous work on the analysis of \ranking~\cite{Karp90},
we consider a set $\U$ of \emph{instances,} each of which has the
form $(\sigma, u)$, where $\sigma$ is a permutation on $V$ and $u$ is a node in $V$.
An instance $(\sigma, u)$ is \emph{good} if the node $u$ is matched
when \ranking is run with $\sigma$, and \emph{bad} otherwise; an event
is a subset of instances.
As argued in~\cite{Matthias12,Goel12}, one can assume that $G$
contains a perfect matching when analyzing the ratio of \ranking.  Hence,
the performance ratio of \ranking is the fraction of good instances.

\noindent \textbf{(1) Relating Bad and Good Events to Form LP Constraints.}
A simple combinatorial argument~\cite{Karp90} is often used to relate bad and good instances.
For example, if each bad instance relates to
to at least two good instances, and
each good instance is related to at most one bad instance, then
the fraction of good instances would be at least $\frac{2}{3}$.  
By considering the structural properties of \ranking, one can define
various relations between different bad and good events,
and hence can generate various constraints in an LP, whose optimal value
gives a lower bound on the performance ratio.

Despite the simplicity
of this combinatorial argument, the analysis of these relations can be elusive
for arbitrary graphs.
Hence,
we define and analyze our relations carefully to 
derive three type of constraints: \emph{monotone} constraints, \emph{evolving} constraints,
and a \emph{boundary} constraint, the last of which involves a novel construction of a
sophisticated relation, and is crucial to the success of our $\LPn$.

\noindent \textbf{(2) Developing New Primal-Dual Techniques for Continuous LP.}
As in previous works, the optimal value of $\LPn$ decreases as $n$ increases.
Hence, to obtain a theoretical proof, 
one needs to analyze the
asymptotic behavior of $\LPn$.
It could be tedious to find the optimal solution of $\LPn$ and investigate its 
limiting behavior.  One could also use experiments (for example using
strongly factor-revealing LP~\cite{Mahdian2011}) to give a proof.  We instead observe that
the $\LPn$ has a continuous $\LPi$ relaxation (in which normal variables becomes a function variable).
However, the monotone constraints in $\LPn$ require that the function
in $\LPi$ be monotonically decreasing.
Moreover, the boundary constraint has
its counterpart in $\LPi$.
To the best of our
knowledge, such continuous LPs have not been analyzed in the literature.  

We describe our formal notation in Section~\ref{sec:pre}.  In Section~\ref{sec:bad-good},
we relate bad and good events in order to form $\LPn$.
In Section~\ref{sec:lowerbound}, we 
prove a lower bound on the performance
ratio; in particular, we
develop new
primal-dual and complementary slackness characterization of a general class of continuous LP, and 
solve the continuous $\LPi$ relaxation (and its dual).
In Appendix~\ref{sec:hardness}, we describe
a hard instance and our experiments show that \ranking performs no better than $0.724$.
\ignore{
We also describe in detail the technical issues concerning previous works~\cite{Matthias12,Goel12}
in Appendix~\ref{appendix:issues}.}

\subsection{Related Work}

We describe and compare the most relevant related work.  Please refer to
the references in~\cite{Matthias12,Goel12} for a more comprehensive background
of the problem.
We describe \gmg general enough so that we can compare 
different works that are studied under
different names and settings. 
Dyer and Frieze~\cite{DyerF91} showed
that picking a permutation of unordered pairs uniformly at random
cannot produce a constant ratio that is  strictly greater than $0.5$.
On the other hand, this framework also includes the \mrg algorithm, which
was analyzed by  by Aronson et al.~\cite{Aronson1995}
to prove the first non-trivial constant performance ratio crossing the 0.5 barrier.
One can also
consider \emph{adaptive} algorithms in which the algorithm
is allowed to change the order in the remaining list
after seeing the probing results; although hardness
results have been proved for adaptive algorithms~\cite{Goel12}, no algorithm
in the literature seems to utilize this feature yet.



\noindent \textbf{On Bipartite Graphs.}
Running \ranking on bipartite graphs for
the \gmg is equivalent to running ranking~\cite{Karp90} 
for the \obm problem with \emph{random arrival order}~\cite{Karande2011}.
From Karande, Mehta and Tripathi~\cite{Karande2011}, one can conclude that
\ranking achieves ratio $0.653$ on bipartite graphs. Moreover, they constructed
a hard instance in which \ranking performs no better than $0.727$;
we modify their hard instance and improve the hardness to $0.724$.

On a high level, most works on analyzing ranking or similar randomized algorithms
on matching are based on variations of the framework by Karp et al.~\cite{Karp90}.
The basic idea is to relate different bad and good events to form 
constraints in an LP, whose asymptotic behavior is analyzed when $n$ is large.
For \obm, Karp et al.~\cite{Karp90} showed that
ranking achieves performance ratio $1 - \frac{1}{e}$;
similarly, Aggarwal et al.~\cite{Aggarwal2011} also showed that a
modified version of ranking achieves the same ratio for the node-weighted version
of the problem.

Sometimes very sophisticated mappings are used to relate
different events, and produce LPs whose asymptotic behavior is difficult to analyze.
Mahdian and Yan~\cite{Mahdian2011} developed the technique of strongly
factor-revealing LP.  The idea is to consider another family of LPs
whose optimal values are all below the asymptotic value of the original LP.
Hence, the optimal value of any LP (usually a large enough instance) in the new family
can be a lower bound on the performance ratio. The results of~\cite{Mahdian2011}
implies that for the \gmg on bipartite graphs, \ranking
achieves performance ratio $0.696$.

\noindent \textbf{Recent Attempts.}
No attempts have been made in the literature to theoretically
improve the $0.5 + \epsilon$ ratio for arbitrary graphs until 
two recent papers appeared in FOCS 2012.
Poloczek and Szegedy~\cite{Matthias12} used a technique known
as \emph{contrast analysis} to analyze 
the \mrg algorithm and gave ratio $\frac{1}{2}+\frac{1}{256} \approx 0.5039$.
\ignore{However, we could not complete their analysis
and give details about the technical issues in Appendix~\ref{appendix:contrast}.  Although
we could not find a way to resolve the issues, it is possible
that their argument can be fixed.}

Goel and Tripathi~\cite{Goel12} showed
a hardness result of $0.7916$ for any algorithm and $0.75$ for adaptive vertex-iterative algorithms.
They also analyzed the \ranking algorithm for a better performance ratio.  
\ignore{Following the same framework, they relate bad and good events to form LP constraints.
However, we also could not complete the structural analysis, and details are in
Appendix~\ref{appendix:proof}.}  Moreover, they used strongly 
factor-revealing LP to analyze
the asymptotic behavior of their LP; we ran experiment on the LP described in their
paper and could not reproduce the ratio $0.56$.  On the contrary,
we discovered that the optimal value of their original LP drops to $0.5001$ when $n = 400$.
Hence, we do not believe strongly factor-revealing LPs can be used to
analyze their original LP to give a ratio larger than $0.5001$.  We describe the
details in Appendix~\ref{appendix:lp}, which includes a link to source codes if the reader
would like to verify our experimental results.

\noindent \textbf{Continuous LP.}
Duality and complementary slackness properties
of continuous LP were
investigated by Tyndall~\cite{Tyndall65} and Levinson~\cite{Levinson66}.
Anand et al.~\cite{AnandGK12} used continuous LP relaxation to analyze
online scheduling.  

%

\section{Preliminaries} \label{sec:pre}

Let $[n] := \{ 1,2,\ldots,n \}$, $[a..b] := \{ a,a+1,\ldots,b \}$ for $1\leq a\leq b$,
and $\Omega$ be the set of all permutations of the nodes in $V$, where each
permutation is a bijection $\sigma : V \rightarrow [n]$.  The \emph{rank} of node $u$
in $\sigma$ is $\sigma(u)$, where smaller rank means higher priority.

\noindent \textbf{The \ranking algorithm.}
For the \gmg, the algorithm selects a permutation $\sigma \in \Omega$ uniformly at
random, and returns a list $L$ of unordered pairs according to the lexicographical order
induced by $\sigma$.  Specifically, given two pairs $e_1$ and $e_2$ (where for
each $i$, $e_i = \{u_i, v_i\}$ and $\sigma(u_i) < \sigma(v_i)$), the pair $e_1$
has higher priority than $e_2$ if (i) $\sigma(u_1) < \sigma(u_2)$, or (ii) $u_1 = u_2$ and
$\sigma(v_1) < \sigma(v_2)$. 
Each pair of nodes in $G(V,E)$ is probed according to the order given by $L$;
initially, all nodes are \emph{unmatched}.
In the round when the pair $e = \{u,v\}$ is probed,
if both nodes are currently \emph{unmatched} and
the edge $e$ is in $E$, 
then each of $u$ and $v$ is \emph{matched},
and they are each other's \emph{partner} in $\sigma$;
moreover, if $\sigma(u) < \sigma(v)$ in this case, we say that $u$ \emph{chooses} $v$.
Otherwise, if at least one of $u$ and $v$ is already matched or
there is no edge between them in $G$, we skip to
the next pair in $L$ until all pairs in $L$ are probed.

After running \ranking with $\sigma$ (or in general probing with list $L$),
we denote the resulting matching by $M(\sigma)$ (or $M(L)$),
and we say that a node is matched in $\sigma$ (or $L$) 
if it is matched in $M(\sigma)$ (or $M(L)$).
Given a probing list $L$, suppose $L_u$ denotes the probing list
obtained by removing all occurrences of $u$ in $L$ such that $u$ always
remains unmatched.  The following lemma is useful.

\begin{lemma}[Removing One Node.] \label{lemma:removal}
The symmetric difference $M(L) \oplus M(L_u)$
is an alternating path, which contains at least one edge
\emph{iff} $u$ is matched in $L$.
\end{lemma}

\begin{proof}
Observe that probing $G$ with $L_u$ is equivalent
to probing $G_u$ with $L$, where $G_u$ is exactly the same as $G$ except that
the node $u$ is labeled \emph{unavailable} and will not be matched in any case.
Hence, we will use the same $L$ to probe $G$ and $G_u$, and compare what happens in each round
to the corresponding matchings $M = M(L)$ and $M_u = M(L_u)$.
For the sake of this proof, ``unavailable'' and ``matched''
are the same \emph{availability status}, while ``unmatched'' is a different
availability status.

We apply induction on the number of rounds of probing.  Observe that
the following invariants hold initially. (i) There is exactly one node known
as the \emph{crucial} node (which is initially $u$) that has 
different availability in $G$ and $G_u$. (ii) The symmetric difference 
$M(L) \oplus M(L_u)$ is an alternating path connecting $u$ to the crucial node;
initially, this path is degenerate.

Consider the inductive step.  Observe that the crucial node and $M(L) \oplus M(L_u)$ 
do not change in a
round except for the case when the pair being probed is an edge in $G$ (and $G_u$),
involving the crucial node $w$ with another currently unmatched 
node $v$ in $G$, and hence
$v$ is also unmatched in $G_u$, as the induction hypothesis states that
every other node apart from the crucial node has the same availability in both graphs.
In this case, this edge is added to exactly one of $M$ and $M_u$.  Therefore, $w$ is
matched in both graphs (so no longer crucial), and
$v$ becomes the new crucial node;
moreover, the edge $\{w,v\}$ is added to $M(L) \oplus M(L_u)$, which now
is a path connecting $u$ to $v$.  This completes the
inductive step.

Observe that $u$ is matched in $M$ in the end, \emph{iff} in some round an edge involving $u$
must be added to $M$ but not to $M_u$, which is equivalent to the case when
 $M \oplus M_u$ contains at least one edge.
\end{proof}

The \emph{performance ratio} $r$ of \ranking on $G$ is
the expected number of nodes matched by the algorithm to the number
of nodes in a maximum matching in $G$, where the randomness comes from the
random permutation in $\Omega$.  
We consider the
set $\U := \Omega \times V$ of \emph{instances}; an \emph{event}
is a subset of instances. An instance $(\sigma,u)\in \U$ is \emph{good} 
if $u$ is matched in $\sigma$, and \emph{bad} otherwise.

{\bf Perfect Matching Assumption.} According to Corollary 2 of~\cite{Matthias12} (and
also implied by our Lemma~\ref{lemma:removal}), 
without loss of generality, we can assume that the graph $G(V,E)$ has a perfect matching $M^*\subseteq E$ that matches all nodes in $V$.
For a node $u$, we denote by
$u^*$ the partner of $u$ in $M^*$ and we call $u^*$ the \emph{perfect partner} of $u$.
From now on, we consider \ranking on such a graph
$G$ without mentioning it explicitly again.  Observe that
for all $\sigma \in \Omega$, $(\sigma, \sigma^{-1}(1))$ is always good;
moreover, the performance ratio is the fraction of good instances.


\begin{definition}[$\sigma_u$, $\sigma_u^i$]
For a permutation $\sigma$, let $\sigma_u$ be the permutation obtained by removing $u$ from $\sigma$ while keeping the relative order of other nodes unchanged;
running \ranking with $\sigma_u$ means running $\sigma$ while keeping $u$ always unavailable (or simply deleting $u$ in $G$).
Let $\sigma_u^i$ be the permutation obtained by inserting $u$ into $\sigma_u$ at rank $i$ and keeping the relative order of other nodes unchanged.
\end{definition}


\begin{fact}[\ranking is Greedy]\label{fact:ranking-greedy} Suppose \ranking is run with permutation $\sigma$. If $u$ is unmatched in $\sigma$, then each neighbor $w$ of $u$ (in $G$) is matched to some node $v$ in $\sigma$ with $\sigma(v)<\sigma(u)$.
\end{fact}

\ignore{
\begin{proof}
Fix any neighbor $w$ of $u$ and notice that $w$ must be matched otherwise \textsc{Ranking} will match $u$ and $w$.
If $w$ chooses $v$ as the second endpoint in $M(\sigma)$, then we know that $\sigma(v)<\sigma(u)$ since at that moment, $\sigma(u)$ is unmatched;
if $v$ chooses $w$ as the second endpoint in $M(\sigma)$, then we also know that $\sigma(v)<\sigma(u)$ since at that moment, $\sigma(u)$ is unmatched while $v$ is considered first.
\end{proof}
}

Similar to~\cite[Lemma 3]{Matthias12}, the following Fact is an easy corollary of Lemma~\ref{lemma:removal},
by observing that if $(\sigma,u)$ is bad, then $M(\sigma) = M(\sigma_u)$.

\begin{fact}[Symmetric Difference]\label{fact:symmetric-difference}
Suppose $(\sigma,u)$ is bad, and 
$(\sigma_u^i,u)$ is good for some $i$. Then, the symmetric difference  $M(\sigma) \oplus M(\sigma_u^i)$ is an alternating path $P$ with at least one edge,
where except for the endpoints of $P$ (of which $u$ is one),
every other node in $G$ is either matched in both $\sigma$ and $\sigma_u^i$, or unmatched
in both.
%
%
\end{fact}

\begin{definition}[$Q_t$, $R_t$ and $S_t$]
For each $t\in[n]$, let $Q_t$ be the good event that the node at rank $t$ is matched, where
$Q_t := \{ (\sigma, u):\sigma\in \Omega, u=\sigma^{-1}(t) \text{ is matched in } \sigma \}$;
similarly, let $R_t$ be the bad event that the node at rank $t$ is unmatched, where
$R_t := \{ (\sigma, u):\sigma\in \Omega, u=\sigma^{-1}(t) \text{ is unmatched in } \sigma \}$.

Moreover, we define the \emph{marginally bad event} $S_t$ at rank $t\in[2..n]$ by
$S_t := \{ (\sigma, u)\in R_t: (\sigma_u^{t-1}, u)\notin R_{t-1} \}$;
observe that $S_1 = R_1 = \emptyset$.

Given any $(\sigma, u) \in \U$, the \emph{marginal position} of $u$ with respect to $\sigma$
is the (unique) rank $t$ such that  $(\sigma_u^t, u) \in S_t$, and is \emph{null} if no such $t$ exists.
\end{definition}

Note that for each $t \in [n]$, $Q_t$ and $R_t$ are disjoint and $|Q_t \cup R_t|=n!$.

\begin{definition}[$x_t$, $\alpha_t$]
For each $t\in[n]$, let $x_t=\frac{|Q_t|}{n!}$ be the probability that a node at rank $t$ is matched, over the random choice of permutation $\sigma$. Similarly, we let $\alpha_t=\frac{|S_t|}{n!}$;
observe that $1 - x_t = \frac{|R_t|}{n!}$.
\end{definition}

Note that the performance ratio is $\frac{1}{n} \sum_{t=1}^n x_t$,
which will be the objective function of our minimization LP.
Observe that all $x_t$'s and $\alpha_t$'s are between $0$ and $1$, and $x_1=1$ and $\alpha_1=0$.
We derive constraints for the variables in the next section.

%



\section{Relating Bad and Good Events to Form LP Constraints} \label{sec:bad-good}

In this section we define some relations between
bad and good events to form LP constraints.  The high level idea
is as follows. Suppose $f$ is a relation between $A$ and $B$, where
$f(a)$ is the set of elements in $B$ related to $a \in A$,
and $f^{-1}(b)$ is the set of elements in $A$ related to $b \in B$.
The \emph{injectivity} of $f$ is the minimum integer $q$ such that
for all $b \in B$, $|f^{-1}(b)| \leq q$.
If $f$ has injectivity $q$,  
we have the inequality $\sum_{a \in A} |f(a)| \leq q |B|$,
which follows from counting the number of edges in the 
bipartite graph induced by $f$ on $A$ and $B$.  In our constructions, usually
calculating $|f(a)|$ is straightforward, but sometimes special attention
is required to bound the injectivity.

\subsection{Monotone Constraints: $x_{t-1} \geq x_t$, $t \in [2..n]$}\label{ssec:basiclemmas}

These constraints follow from Lemma~\ref{lemma:bad-marginalbad}
 as the $\alpha_t$'s are non-negative.
%
\begin{lemma}[Bad-to-Marginally Bad] \label{lemma:bad-marginalbad}
For all $t\in[n]$, we have $1-x_t=\sum_{i=1}^t \alpha_i$; this implies that
for $t \in [2..n]$, $x_{t-1}-x_t = \alpha_t$.
\end{lemma}

\begin{proof}
Fix $t \in [n]$. From the definitions of $x_t$ and $\alpha_t$, it suffices to provide a bijection $f$ from $R_t$ to $\cup_{i=1}^t S_i$.
Suppose $(\sigma,u)\in R_t$.  This means $(\sigma, u)$ is bad, and
hence $u$ has a marginal position $t_u \leq t$ with respect to $\sigma$.
We define $f(\sigma, u) := (\sigma_u^{t_u},u) \in \cup_{i=1}^t S_i$.

\emph{Surjective:} for each $(\rho, v)\in \cup_{i=1}^t S_i$,
the marginal position of $v$ with respect to $\rho$ is some $i \leq t$;
hence, it follows that $(\rho_v^t, v) \in R_t$ is bad, and
we have $f(\rho_v^t,v) = (\rho, v)$.  

\emph{Injective:}
if we have $f(\sigma, u) = (\rho, v)$, it must be the case that
$u= v$, $\sigma(u) = t$, and $\rho = \sigma_u^i$ for some $i$; this implies
that $\sigma$ must be $\rho_v^t$.

Hence, $|R_t|=|\cup_{i=1}^t S_i|=\sum_{i=1}^t |S_i|$, which is equivalent to $1-x_t=\sum_{i=1}^t \alpha_i$, if we divide the equality by $n!$ on both sides.
\end{proof}

\subsection{Evolving Constraints: $\left(1-\frac{t-1}{n} \right) x_t + \frac{2}{n} \sum_{i=1}^{t-1} x_i \geq 1$, $t \in [2..n]$}\label{ssec:movingdown}

The monotone constraints require that the $x_t$'s do not
increase.  We next derive the \emph{evolving}
constraints that prevent the $x_t$'s from dropping too fast.
Fix $t \in [2..n]$.  We shall define a relation $f$
between $\cup_{i=1}^t S_i$ and $\cup_{i=1}^{t-1} Q_i$
such that $f$ has injectivity $1$, and 
for $(\sigma, u) \in S_i$, $|f(\sigma,u)| = n - i + 1$.
This implies Lemma~\ref{lemma:1-to-n-t+1};
from Lemma~\ref{lemma:bad-marginalbad}, we can express
$\alpha_i = x_{i-1}-x_i$ (recall $\alpha_1=0$), and rearrange the terms to obtain the required
constraint.


\begin{lemma}[$1$-to-($n-i+1$) Mapping]\label{lemma:1-to-n-t+1}
For all $t \in [2..n]$, we have $\sum_{i=1}^{t} (n-i+1) \alpha_i \leq \sum_{i=1}^{t-1} x_i$.
\end{lemma}

\begin{proof}
We define a relation $f$ between $A := \cup_{i=1}^t S_i$ and $B := \cup_{i=1}^{t-1} Q_i$.
Let $(\sigma, u) \in A$ be a marginally bad instance.
Then, there exists a unique $i \in [2..t]$
such that
$(\sigma, u) \in S_i$. If we move $u$ to any position $j \in [i..n]$,
$(\sigma_u^j, u)$ is still bad, because $i$ is the marginal
position of $u$ with respect to $\sigma$.  Moreover, 
observe that $M(\sigma_u) = M(\sigma) =M(\sigma_u^j)$ for all $j \in [i..n]$.
Hence, it follows that for all $j \in [i..n]$, node $u$'s perfect partner
$u^*$ is matched in $\sigma_u^j$ to the 
same node $v$ such that $\sigma(v) = \sigma_u^j(v) \leq i-1 \leq t-1$,
where the first inequality follows from Fact~\ref{fact:ranking-greedy}.
In this case, we define $f(\sigma,u) := \{(\sigma_u^j, v): j \in [i..n]\} \subset B$,
and it is immediate that $|f(\sigma, u)| = n - i + 1$.

%
%
%
%



\noindent \emph{Injectivity.}
Suppose $(\rho, v) \in B$ is related to some $(\sigma, u) \in A$.
It follows that $v$ must be matched to $u^*$ in $\rho$; hence, $u$ is uniquely
determined by $(\rho,v)$. Moreover, $(\rho, u)$ must be bad, and suppose
the marginal position of $u$ with respect to $\rho$ is $i$, which is also uniquely determined.
Then, it follows that $\sigma$ must be $\rho_u^i$.
Hence, $(\rho,v)$ can be related to at most one element in $A$.

Observing that $S_1 = \emptyset$,
the result follows from $\sum_{i=1}^t (n - i +1) |S_i| = \sum_{a \in A} |f(a)| \leq |B| = \sum_{i=1}^{t-1} |Q_i|$, since $|S_i| = n! \alpha_i$ and $|Q_i| = n! x_i$.
\end{proof}

\subsection{Boundary Constraint: $x_n + \frac{3}{2n} \sum_{i=1}^n x_i \geq 1$}\label{ssec:movingaround}

According to experiments, the monotone and the evolving constraints alone
cannot give ratio better than 0.5.  The \emph{boundary constraint} is crucial
to the success of our LP, and hence we analyze our construction carefully.  
The high level idea is that we define a relation $f$
between $R_n$ and $Q := \cup_{i=1}^n Q_i$.  As we shall see,
it will be straightforward to show that $|f(a)| = 2n$ for
each $a \in R_n$, but it will require some work to show that
the injectivity is at most 3.  Once we have established these results,
the boundary constraint follows immediately
from $\sum_{a \in R_n} |f(a)| \leq 3 |Q|$,
 because $\frac{|R_n|}{n!} = 1 - x_n$
and $\frac{|Q_i|}{n!} = x_i$.

\noindent \textbf{Defining relation $f$ between $R_n$ and $Q$.}
Consider a bad instance $(\sigma,u) \in R_n$.  We define $f(\sigma, u)$
such that for each $i \in [n]$, $(\sigma,u)$ \emph{produces}
exactly two good instances of the form $(\sigma_u^i, *)$.

\begin{tabular}{|p{\textwidth}|}
\hline
\vspace{0.1cm}
For each $i\in[n]$, we consider $\sigma_u^i$:
\begin{enumerate}
\item if $u$ is unmatched in $\sigma_u^i$: \texttt{($u$ and $u^*$ cannot be both unmatched)}

$\textsf{R(1)}$: produce $(\sigma_u^i,u^*)$ and include it in $f(\sigma,u)$;

$\textsf{R(2)}$: let $v$ be the partner of $u^*$ in $\sigma_u^i$; produce $(\sigma_u^i,v)$ and include it in $f(\sigma,u)$.

\item if $u$ is matched in $\sigma_u^i$:

$\textsf{R(3)}$: produce $(\sigma_u^i,u)$ and include it in $f(\sigma,u)$;
\begin{enumerate}
\item if $u^*$ is matched to $u$ in $\sigma_u^i$:

$\textsf{R(4)}$: produce $(\sigma_u^i,u^*)$ and include it in $f(\sigma,u)$;

\item if $u^*$ is matched to $v\neq u$ in $\sigma_u^i$:

$\textsf{R(5)}$: produce $(\sigma_u^i,v)$ and include it in $f(\sigma,u)$;

\item if $u^*$ is unmatched in $\sigma_u^i$: \texttt{(all neighbors of $u^*$ in $G$ must be matched)}

$\textsf{R(6)}$: let $v_o$ be the partner of $u^*$ in $\sigma$, produce $(\sigma_u^i,v_o)$ and include it in $f(\sigma,u)$.
\end{enumerate}
\end{enumerate} \\
\hline
\end{tabular}

Observe that for $i \in [6]$, applying each rule $\textsf{R(i)}$ produces exactly one good instance.  Moreover, for each $i\in [n]$, when
we consider $\sigma_u^i$, exactly $2$ rules will be applied: if $u$ is unmatched in $\sigma_u^i$, then $\textsf{R(1)}$ and $\textsf{R(2)}$ will be applied; if $u$ is matched in $\sigma_u^i$, then $\textsf{R(3)}$ and one of $\{\textsf{R(4)}, \textsf{R(5)}, \textsf{R(6)}\}$ will be applied.

\begin{observation}
For each $(\sigma, u)\in R_n$, we have $|f(\sigma,u)|=2n$.
\end{observation}

\begin{observation} \label{obs:unique_u}
If $(\rho,x) \in f(\sigma,u)$,
then $\sigma = \rho_u^n$ and exactly one rule can be applied
to $(\sigma, u)$ to produce $(\rho,x)$.
\end{observation}

\noindent \textbf{Bounding Injectivity.} We first show that
different bad instances in $R_n$ cannot produce the same good instance using the same rule.

\begin{lemma}[Rule Disjunction] \label{lemma:rule-disjoint}
For each $i \in [6]$, any $(\rho,x) \in Q$
can be produced by at most one $(\sigma, u) \in R_n$ using $\textsf{R(i)}$.
\end{lemma}

\begin{proof}
Suppose $(\rho,x) \in Q$ is produced using a particular rule $\textsf{R(i)}$
by some $(\sigma, u) \in R_n$.  We wish to show that in each case $i \in [6]$,
we can recover $u$ uniquely, in which case $\sigma$ must be $\rho_u^n$.

The first 5 cases are simple. Let $y$ be the partner of $x$ in $\rho$.
If $i=1$ or $i=4$, we know that $x=u^*$ and hence we can recover $u=x^*$;
if $i=2$ or $i=5$, we know that $y=u^*$ and hence we can recover $u=y^*$;
if $i=3$, we know that $u=x$.

For the case when $i=6$, we need to do a more careful analysis.
Suppose $\textsf{R(6)}$ is applied to $(\sigma,u) \in R_n$ to produce $(\rho,x)$.
Then, we can conclude the following:
(i) in $\sigma = \rho_u^n$, $u$ is unmatched, and $u^*$ is matched to $x$; 
(ii) in $\rho$, $u$ is matched, $u^*$ is unmatched, and $x$ is matched.

For contradiction's sake, assume that $u$ is not unique and there are two
$u_1 \neq u_2$ that satisfy the above properties.
It follows that
$u_1^*\neq u_2^*$ and according to property (ii),
in $\rho$, both $u_1$ and $u_2$ are matched, and
 both $u_1^*$ and $u_2^*$ are unmatched;
hence, all 4 nodes are distinct.
Without loss of generality, we assume that $\rho(u_1^*) < \rho(u_2^*)$.
Let $\sigma_2 := \rho_{u_2}^n$, and observe that
$\sigma_2(u_1^*) < \sigma_2(u_2^*)$.

Now, suppose we start with $\sigma_2$,
and consider what happens when $u_2$ is promoted in $\sigma_2$ resulting in $\rho$.
Observe that $u_2$ changes from unmatched in $\sigma_2$ to matched
in $\rho$, and by property (i), $u_2^*$ changes
from matched in $\sigma_2$ to unmatched in $\rho$.
From Fact~\ref{fact:symmetric-difference}, every other node
must remain matched or unmatched in both $\sigma_2$ and $\rho$;
in particular, $u_1^*$ is unmatched in $\sigma_2$.
However, $x$ is a neighbor of both $u_1^*$ and $u_2^*$ (in $G$), and
$\sigma_2(u_1^*) < \sigma_2(u_2^*)$, but $x$ is matched to
$u_2^*$ in $\sigma_2$; this contradicts Fact~\ref{fact:ranking-greedy}.
\end{proof}

Lemma~\ref{lemma:rule-disjoint} immediately implies that
the injectivity of $f$ is at most 6.  However, to show a better bound of 3,
we need to show that some of the rules cannot be simultaneously applied to produce
the same good instance $(\rho, x)$.
We consider two cases for the remaining analysis.

\noindent \textbf{Case (1): $x$ is matched to $x^*$ in $\rho$}

\begin{lemma}\label{lemma:x-to-x*}
For $(\rho,x)\in Q$, if $x$ is matched to $x^*$ in $\rho$, then we have $|f^{-1}(\rho,x)|\leq 3$.
\end{lemma}

\begin{proof}
If $(\rho,x)$ is produced using $\textsf{R(1)}$, then $x^*$ must be unmatched in $\rho$;
if $(\rho,x)$ is produced by $(\sigma,u)$ using $\textsf{R(2)}$, then $x$ must be matched to $u^* \, (\neq x^*)$ in $\rho$ since $x\neq u$; similarly, 
if $(\rho,x)$ is produced by $(\sigma,u)$ using $\textsf{R(5)}$, then $x \, (\neq u)$ must be matched to $u^* \, (\neq x^*)$ in $\rho$.

Hence, $(\rho,x)$ cannot be produced by $\textsf{R(1)}$, $\textsf{R(2)}$ or $\textsf{R(5)}$,
and at most three remaining rules can produce it.
It follows from Lemma~\ref{lemma:rule-disjoint} that
 $|f^{-1}(\rho,x)|\leq 3$.
\end{proof}

\noindent \textbf{Case (2): $x$ is not matched to $x^*$ in $\rho$}

\begin{observation}[Unused Rule]\label{observation:unused-rule}
For $(\rho,x)\in Q$, if $x$ is not matched to $x^*$ in $\rho$, then $(\rho,x)$ cannot
be produced by applying $\textsf{R(4)}$.
\end{observation}

Out of the remaining 5 rules,
we show that $(\rho,x)$ can be produced from at most one of $\{\textsf{R(2)},\textsf{R(5)}\}$, 
and at most two of $\{\textsf{R(1)},\textsf{R(3)},\textsf{R(6)}\}$.  After
we show these two lemmas, we can immediately conclude from Lemma~\ref{lemma:rule-disjoint} that
$|f^{-1}(\rho,x)|\leq 3$ and complete the case analysis.


\begin{lemma}[One in $\{\textsf{R(2)},\textsf{R(5)}\}$]\label{lemma:combine-rule-2-and-5}
Each $(\rho,x) \in Q$ cannot be produced from both \textsf{R(2)} and \textsf{R(5)}.
\end{lemma}

\begin{proof}
Suppose the opposite is true:
$(\sigma_1, u_1)$ produces $(\rho,x)$ according to $\textsf{R(2)}$,
and $(\sigma_2, u_2)$ produces $(\rho,x)$ according to $\textsf{R(5)}$.
This implies that in $\rho$, $x$ is matched to both $u_1^*$ and $u_2^*$,
which means $u_1 = u_2$.
By Observation~\ref{obs:unique_u}, this means $\sigma_1 = \sigma_2$,
which contradicts the fact that
the same $(\sigma, u) \in R_n$ cannot use two different rules
to produce the same $(\rho,x) \in Q$.
\end{proof}

\begin{lemma}[Two in $\{\textsf{R(1)},\textsf{R(3)},\textsf{R(6)}\}$]\label{lemma:2-out-of-3}
Each $(\rho,x)\in Q$ cannot be produced from all three of $\textsf{R(1)}$,
$\textsf{R(3)}$ and $\textsf{R(6)}$.
\end{lemma}

\begin{proof}
Assume the opposite is true.
Suppose $(\sigma_1,u_1)$ produces $(\rho,x)$ using $\textsf{R(1)}$; then, $x=u_1^*$ (hence, $x$ is
a neighbor of $u_1$ in $G$) and $u_1$ is unmatched in $\rho$.
Suppose $(\sigma_2,u_2)$ produces $(\rho,x)$ using $\textsf{R(3)}$; then,
 $x=u_2$ is unmatched in $\sigma_2$, and matched in $\rho$.
Suppose $(\sigma_3,u_3)$ produces $(\rho,x)$ using $\textsf{R(6)}$; then, $u_3$ is matched in $\rho$, $u_3^*$ is unmatched in $\rho$ and $x$ is a neighbor (in $G$) of $u_3^*$.

By Observation~\ref{obs:unique_u}, all of $u_1$, $u_2$ and $u_3$ are distinct.
In particular, observe that $u_1=x^*=u_2^* \neq u_3^*$; hence,
all of $u_1$, $u_2$ and $u_3^*$ are distinct (since $u_2$ is matched in $\rho$, but
the other two are not).

Now, suppose we start from $\sigma_2 = \rho_x^n$ and promote
$x = u_2$ resulting in $\rho$.  Observe that $u_2$ changes from
unmatched in $\sigma_2$ to matched in $\rho$,
and both $u_1$ and $u_3^*$ are unmatched in $\rho$.
By Fact~\ref{fact:symmetric-difference}, at least one of
$u_1$ and $u_3^*$ is unmatched in $\sigma_2$; however,
both $u_1$ and $u_3^*$ are neighbors of $x = u_2$ (in $G$), which is unmatched in $\sigma_2$.
This contradicts that fact that in any permutation, 
two unmatched nodes cannot be neighbors in $G$.
\end{proof}

We have finally finished the case analysis, and can conclude the
$f$ has injectivity at most 3, thereby achieving the boundary constraint.

%
%
%
%

\subsection{Lower Bound the Performance Ratio by LP Formulation}
\label{sec:finite_lp}

Combining all the proved constraints, the following $\LPn$ gives a lower bound
on the performance ratio when \ranking is run on a graph with $n$ nodes.
It is not surprising that the optimal value of $\LPn$ decreases as $n$ increases 
(although our proof does not rely on this).
In Section~\ref{sec:lowerbound}, we analyze the continuous relaxation $\LPi$
in order to give a lower bound for all finite $\LPn$, thereby proving a lower bound
on the performance ratio of \ranking.
\begin{align*}
\LPn \qquad\qquad \min \qquad \txts \frac{1}{n} \txts\sum_{t=1}^{n} x_t & \\
\text{s.t.} \qquad x_1 = 1, & \\
x_{t-1} - x_t \geq 0, & \qquad  t \in [2..n] \\
\txts\left(1-\frac{t-1}{n} \right) x_t + \frac{2}{n} \sum_{i=1}^{t-1} x_i \geq 1, & \qquad t \in [2..n] \\
x_n + \txts \frac{3}{2n} \sum_{t=1}^{n} x_t \geq 1, & \\
x_t \geq 0, & \qquad t \in [n].
\end{align*}

\section{Analyzing $\LPn$ via Continuous $\LPi$ Relaxation}\label{sec:lowerbound}

In this section, we analyze the limiting behavior of $\LPn$
by solving its continuous $\LPi$ relaxation,
which contains both monotone and boundary condition constraints.
We develop new duality and complementary slackness characterizations
to solve for the optimal value of $\LPi$, thereby giving a lower bound
on the performance ratio of \ranking.

%

\subsection{Continuous LP Relaxation}

To form a continuous linear program $\LPi$ from $\LPn$, we replace 
the variables $x_t$'s with 
a function variable $z$ that is differentiable almost everywhere in $[0,1]$. 
The dual $\LDi$ contains a real variable $\gamma$, and function variables
$w$ and $y$, where $y$ is differentiable 
almost everywhere in $[0,1]$.
In the rest of this paper, we use ``$\forall \theta$'' to denote ``for almost all $\theta$'', which means for all but a measure zero set.
 
It is not hard to see that $x_i$ corresponds to $z(\frac{i}{n})$, but
perhaps it is less obvious how $\LDi$ is formed.  We
remark that one could consider the limiting behavior
of the dual of $\LPn$ to conclude that $\LDi$ is the resulting program.
We show in Section~\ref{sec:pd_cont}  
that the pair $(\LPi, \LDi)$ is actually a special case of a more general class
of primal-dual continuous LP.
However, we first show in Lemma~\ref{lemma:lp_relax} that $\LPi$ is a relaxation of $\LPn$.

\vspace{-0.8cm}
\begin{minipage}[t]{0.5\textwidth}
\begin{align*}
& \qquad \LPi \\
\min \qquad & \txts \int_{0}^{1} z(\theta) d \theta \\
\text{s.t.}\qquad & z(0) = 1 \\
& z'(\theta) \leq 0, \, \forall \theta \in [0,1] \\
& \txts (1-\theta) z(\theta) + 2\int_{0}^{\theta} z(\lambda) d \lambda \geq 1, \, \forall \theta\in [0,1] \\
& \txts z(1) + \frac{3}{2} \int_{0}^{1} z(\theta) d \theta \geq 1 \\
& z(\theta) \geq 0, \, \forall \theta\in[0,1].
\end{align*}
\end{minipage}
\begin{minipage}[t]{0.5\textwidth}
\begin{align*}
& \qquad \LDi \\
\max \qquad & \txts \int_{0}^{1} w(\theta) d \theta + \gamma - y(0) \\
\text{s.t.} \qquad & \txts (1-\theta) w(\theta) + 2 \int_{\theta}^{1} w(\lambda) d \lambda \\
& \qquad + \txts \frac{3\gamma}{2} + y'(\theta) \leq 1, \, \forall \theta \in [0,1] \\
& \gamma - y(1) \leq 0 \\
& \gamma, y(\theta), w(\theta) \geq 0, \, \forall \theta \in [0,1]. \\
\end{align*}
\end{minipage}


\begin{lemma} [Continuous LP Relaxation] \label{lemma:lp_relax}
For every feasible solution $x$ in $\LPn$, there exists a feasible solution $z$ in $\LPi$ such that $\int_{0}^{1} z(\theta) d \theta = \frac{1}{n} \sum_{t=1}^{n} x_t$. In particular, the optimal value of $\LPn$ is at least the optimal value of $\LPi$.
\end{lemma}

\begin{proof}
Suppose $x$ is a feasible solution to $\LPn$. Define a step function $z$ in interval $[0,1]$ as follows: $z(0) := 1$ and $z(\theta) := x_t$ for $\theta \in \left(\frac{t-1}{n}, \frac{t}{n}\right]$ and $t\in [n]$. It follows that
\begin{align*}
\txts\int_{0}^{1} z(\theta) d \theta
= \sum_{t=1}^{n} \int_{\frac{t-1}{n}}^{\frac{t}{n}} z(\theta)  d\theta
= \frac{1}{n} \sum_{t=1}^{n} x_t.
\end{align*}
We now prove that $z$ is feasible in $\LPi$. Clearly $z(0) = 1$ and $z'(\theta) = 0$ for $\theta\in [0,1]\setminus \{\frac{t}{n}: 0\leq t\leq n, t\in \Z \}$. For every $\theta \in (0,1]$, suppose $\theta \in \left(\frac{t-1}{n}, \frac{t}{n}\right]$, and we have
\begin{align*}
\txts (1-\theta) z(\theta) + 2\int_{0}^{\theta} z(\lambda) d \lambda
& = \txts (1-\theta) x_t + 2 \sum_{i=1}^{t-1} \int_{\frac{i-1}{n}}^{\frac{i}{n}} z(\theta) d \theta
+ 2 \int_{\frac{t-1}{n}}^{\theta} z(\theta) d \theta \\
& = \txts (1-\theta) x_t + \frac{2}{n} \sum_{i=1}^{t-1} x_i + 2 \left(\theta - \frac{t-1}{n}\right) x_t \\
& = \txts (1 - \frac{t-1}{n} + ( \theta - \frac{t-1}{n})) x_t + \frac{2}{n} \sum_{i=1}^{t-1} x_i  \\
& \geq \txts \left(1-\frac{t-1}{n} \right) x_t + \frac{2}{n} \sum_{i=1}^{t-1} x_i \\
& \geq 1,
\end{align*}
where the last inequality follows from the feasibility of $x$ in $\LPn$. The above inequality holds trivially at $\theta = 0$. For the last constraint, using the fact that $\int_{0}^{1} z(\theta) d \theta = \frac{1}{n} \sum_{t=1}^{n} x_t$ we have
\begin{align*}
\txts z(1) + \frac{3}{2} \int_{0}^{1} z(\theta) d \theta = x_n + \frac{3}{2n} \sum_{t=1}^{n} x_t \geq 1,
\end{align*}
where the last inequality follows from the feasibility of $x$ in $\LPn$.
\end{proof}

\subsection{Primal-Dual for a General Class of Continuous LP}
\label{sec:pd_cont}

We study a class of continuous linear program $\CP$ that 
includes $\LPi$ as a special case. In particular, $\CP$ contains 
monotone and boundary conditions as constraints. Let $K, L>0$ be two real constants. 
Let $A$, $B$, $C$, $F$ be measurable functions on $[0,1]$. 
Let $D$ be a non-negative measurable function on $[0,1]^2$. 
We describe $\CP$ and its dual $\CD$, following which we present 
weak duality and complementary slackness conditions. In $\CP$, the variable is 
a function $z$ that is differentiable almost everywhere in $[0,1]$; in $\CD$, the variables are
a real number $\gamma$, and
 measurable functions $w$ and $y$, where $y$ is differentiable almost 
everywhere in $[0,1]$.

\vspace{-0.8cm}
\begin{minipage}[t]{0.425\textwidth}
\begin{align}
& \qquad \CP \nonumber \\
\min \qquad & \txts p(z) = \int_{0}^{1} A(\theta) z(\theta) d \theta \nonumber \\
\text{s.t.} \qquad & z(0) = K \label{eq:pcon0} \\
& z'(\theta) \leq 0, \, \forall \theta \in [0,1] \label{eq:pcon1} \\
& \txts B(\theta) z(\theta) + \int_{0}^{\theta} D(\theta, \lambda) z(\lambda) d \lambda \nonumber \\
& \qquad \geq C(\theta), \, \forall \theta \in [0,1] \label{eq:pcon2} \\
& \txts z(1) + \int_{0}^{1} F(\theta) z(\theta) d \theta \geq L \label{eq:pcon3} \\
& z(\theta) \geq 0, \, \forall \theta \in [0,1]. \nonumber
\end{align}
\end{minipage}
\begin{minipage}[t]{0.575\textwidth}
\begin{align}
& \qquad \CD \nonumber \\
\max \qquad & \txts d(w, y, \gamma) = \int_{0}^{1} C(\theta) w(\theta) d \theta + L \gamma - K y(0) \nonumber \\
\text{s.t.} \qquad & \txts B(\theta) w(\theta) + \int_{\theta}^{1} D(\lambda, \theta) w(\lambda) d \lambda \nonumber \\
& \qquad + F(\theta) \gamma + y'(\theta) \leq A(\theta), \, \forall \theta \in [0,1] \label{eq:dcon1} \\
& \gamma - y(1) \leq 0 \label{eq:dcon2} \\
& \gamma, y(\theta), w(\theta) \geq 0, \, \forall \theta \in [0,1]. \nonumber
\end{align}
\end{minipage}

\begin{lemma} [Weak Duality and Complementary Slackness] \label{lemma:duality}
Suppose $z$ and $(w,y,\gamma)$ are feasible solutions to $\CP$ and $\CD$ respectively. Then, $d(w,y,\gamma) \leq p(z)$. Moreover, suppose $z$ and $(w,y,\gamma)$ satisfy the following complementary slackness conditions:
\begin{align}
z'(\theta) y(\theta) & = 0, \qquad \forall \theta \in [0, 1] \label{eq:cs1} \\
\txts \left[B(\theta) z(\theta) + \int_{0}^{\theta} D(\theta, \lambda) z(\lambda) d \lambda - C(\theta)\right] w(\theta)&  = 0, \qquad \forall \theta \in [0,1] \label{eq:cs2} \\
\txts \left[ z(1) + \int_{0}^{1} F(\theta) z(\theta) d \theta - L \right] \gamma & = 0 \label{eq:cs3} \\
\txts \left[ B(\theta) w(\theta) + \int_{\theta}^{1} D(\lambda, \theta) w(\lambda) d \lambda + F(\theta) \gamma + y'(\theta) - A(\theta) \right] z(\theta) & = 0, \qquad \forall \theta \in [0,1] \label{eq:cs4} \\
(\gamma-y(1)) z(1) & = 0. \label{eq:cs5}
\end{align}
Then, $z$ and $(w,y,\gamma)$ are optimal to $\CP$ and $\CD$, respectively, and achieve
the same optimal value.
\end{lemma}

\begin{proof}
Using the primal and dual constraints, we obtain
\begin{flalign*}
d(w, y, \gamma) & = \txts \int_{0}^{1} C(\theta) w(\theta) d \theta + L \gamma - K y(0) \\
& \leq \txts \int_{0}^{1} \left[B(\theta) z(\theta) + \int_{0}^{\theta} D(\theta, \lambda) z(\lambda) d \lambda \right] w(\theta) d \theta + L \gamma - K y(0) & \text{by~\eqref{eq:pcon2}} \\
& = \txts \int_{0}^{1} \left[ B(\theta) w(\theta) + \int_{\theta}^{1} D(\lambda, \theta) w(\lambda) d \lambda \right] z(\theta) d\theta + L \gamma - K y(0) &  \text{(*)}\\
& \leq \txts \int_{0}^{1} \left[ A(\theta) - F(\theta) \gamma - y'(\theta) \right] z(\theta) d\theta + L \gamma - K y(0) & \text{by~\eqref{eq:dcon1}} \\
& = \txts \int_{0}^{1} A(\theta) z(\theta) d\theta - \int_{0}^{1} y'(\theta) z(\theta) d\theta
+ \left[L-\int_{0}^{1} F(\theta) z(\theta) d\theta \right] \gamma - K y(0) \\
& \leq \txts \int_{0}^{1} A(\theta) z(\theta) d\theta - \int_{0}^{1} y'(\theta) z(\theta) d\theta
+ z(1) \gamma - K y(0) & \text{by~\eqref{eq:pcon3}} \\
& = \txts \int_{0}^{1} A(\theta) z(\theta) d\theta - y(1)z(1) + y(0)z(0) + \int_{0}^{1} z'(\theta) y(\theta) d\theta
+ z(1) \gamma - K y(0) & \text{(**)} \\
& \leq \txts \int_{0}^{1} A(\theta) z(\theta) d\theta + (\gamma-y(1)) z(1) & \text{by~\eqref{eq:pcon0},~\eqref{eq:pcon1}} \\
& \leq \txts \int_{0}^{1} A(\theta) z(\theta) d\theta & \text{by~\eqref{eq:dcon2}} \\
& = p(z),
\end{flalign*}

where in (*) we change the order of integration by using
Tonelli's Theorem on non-negative measurable function $g$:
$\int_0^1 \int_0^\theta g(\theta, \lambda) d \lambda d \theta = \int_0^1 \int_\theta^1 g(\lambda, \theta) d \lambda d \theta$; and in (**) we use integration by parts.
Moreover, if $z$ and $(w,y,\gamma)$ satisfy conditions~\eqref{eq:cs1} --~\eqref{eq:cs5}, then all the inequalities above hold with equality. Hence, $d(w, y, \gamma) = p(z)$; so $z$ and $(w,y,\gamma)$ are optimal to $\CP$ and $\CD$, respectively.
\end{proof}


\subsection{Lower Bound for the Performance Ratio}
The performance ratio of \ranking is lower bounded by the optimal value of $\LPi$. We analyze this optimal value by applying the primal-dual method to $\LPi$. In particular, we construct a primal feasible solution $z$ and a dual feasible solution $(w,y,\gamma)$ that satisfy the complementary slackness conditions presented in Lemma~\ref{lemma:duality}.
Note that $\LPi$ and $\LDi$ are achieved from $\CP$ and $\CD$ by setting $K := 1$, $L := 1$, $A(\theta) := 1$, $B(\theta) := 1-\theta$, $C(\theta) := 1$,
$D(\theta) := 2$, $F(\theta) := \frac{3}{2}$.

We give some intuition on how $z$ is constructed. An optimal solution to $\LPi$ should satisfy the primal constraints with equality for some $\theta$. Setting the constraint $(1-\theta) z(\theta) + 2\int_{0}^{\theta} z(\lambda) d \lambda \geq 1$ to equality we get $z(\theta) = 1 - \theta$. However this function violates the last constraint $z(1) + \frac{3}{2} \int_{0}^{1} z(\theta) d \theta \geq 1$. Since $z$ is decreasing, we need to balance between $z(1)$ and $\int_{0}^{1} z(\theta) d\theta$. 

The intuition is that we set $z(\theta) := 1-\theta$ for $\theta \in [0, \mu]$
and allow $z$ to decrease until $\theta$ reaches some value $\mu \in (0,1)$,
and then $z(\theta) := 1-\mu$ stays constant for $\theta\in[\mu,1]$. To determine the value of $\mu$, note that the equation $z(1) + \frac{3}{2} \int_{0}^{1} z(\theta) d \theta = 1$ should be satisfied, since otherwise we could construct a feasible solution with smaller objective value by decreasing the value of $z(\theta)$ for $\theta\in (\mu,1]$. It follows that $(1- \mu) + \frac{3}{2} \left(1-\mu+\frac{\mu^2}{2} \right) = 1$, that is, the value of $\mu \in (0,1)$ is determined by the equation $3 \mu^2 - 10 \mu + 6 = 0$.

After setting $z$, we construct $(w,y,\gamma)$ carefully to fit the complementary slackness conditions. Formally, we set $z$ and $(w,y,\gamma)$ as follows with their graphs on the right hand side:

\begin{minipage}{0.6\textwidth}
\vspace{-30pt}
\begin{align*}
z(\theta) & =
\begin{cases}
1 - \theta, & 0 \leq \theta \leq \mu \\
1 - \mu, & \mu < \theta \leq 1
\end{cases} \\
w(\theta) & =
\begin{cases}
\frac{2(1-\mu)^2}{(5-3\mu)(1-\theta)^{3}}, & 0 \leq \theta \leq \mu \\
0, & \mu < \theta \leq 1
\end{cases} \\
y(\theta) & =
\begin{cases}
0, & 0 \leq \theta \leq \mu \\
\frac{2(\theta-\mu)}{5-3\mu}, & \mu < \theta \leq 1
\end{cases} \\
\gamma & = \txts \frac{2(1-\mu)}{5-3\mu},
\end{align*}
where $\mu = \frac{5-\sqrt{7}}{3}$ is a root of the equation
\begin{align*}
3 \mu^2 - 10 \mu + 6 = 0.
\end{align*}
\end{minipage}
\begin{minipage}{0.4\textwidth}
\centering
\vspace{-10pt}
\includegraphics[width=0.9\textwidth]{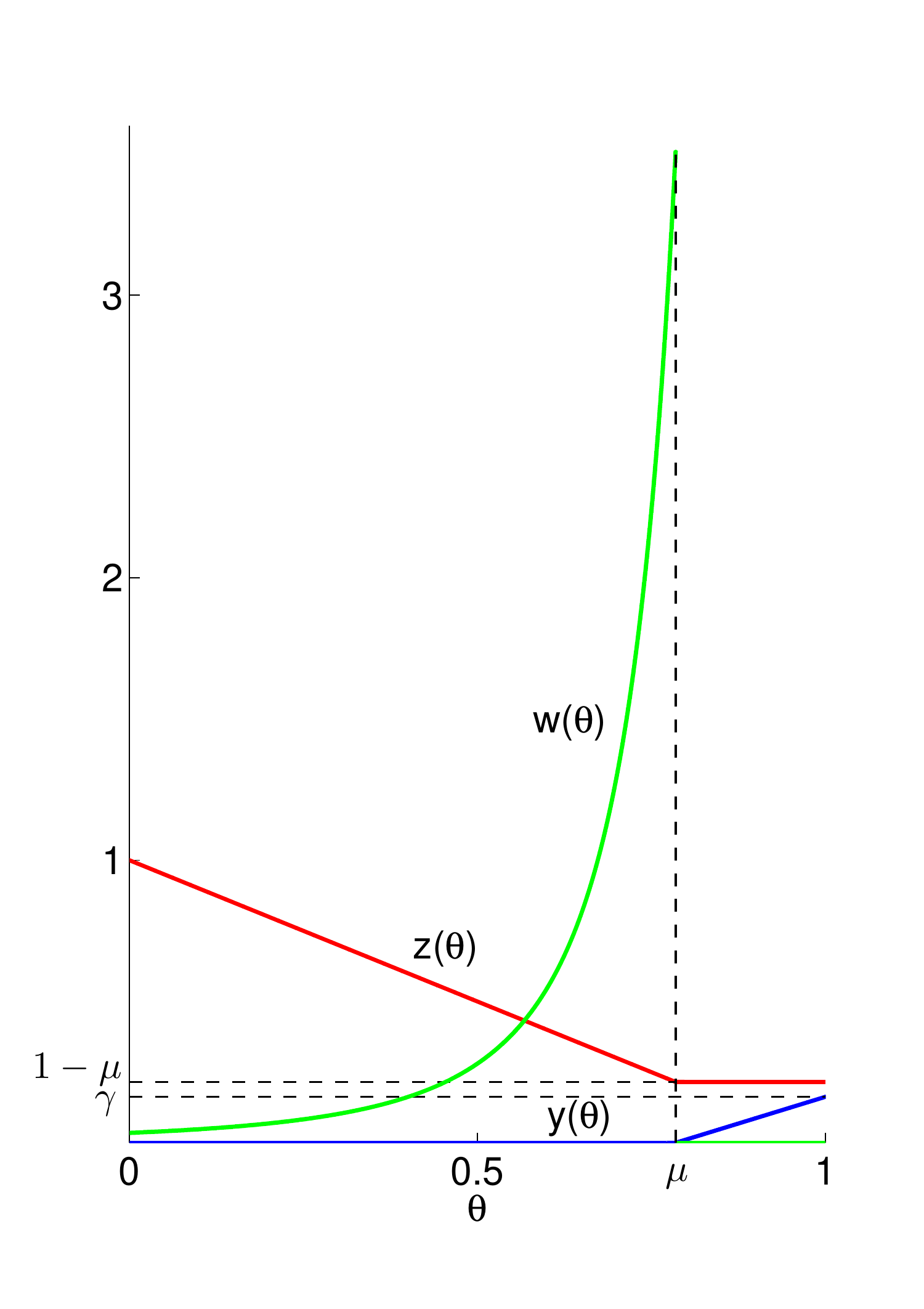}
\vspace{-10pt}
\captionof{figure}{Optimal $z$ and $(w,y,\gamma)$}
\vspace{20pt}
\label{fig:zwy}
\end{minipage}

%
%
%
%



\begin{lemma} [Optimality of $z$ and $(w,y,\gamma)$] \label{lemma:opt}
The solutions $z$ and $(w,y,\gamma)$ constructed above are optimal to $\LPi$ and $\LDi$, respectively. In particular, the optimal value of $\LPi$ is $\frac{2(5-\sqrt{7})}{9}\approx 0.523$.
\end{lemma}
\begin{proof}
We list the complementary slackness conditions and check that they are satisfied by $z$ and $(w,y,\gamma)$. Then Lemma~\ref{lemma:duality} gives the optimality of $z$ and $(w,y,\gamma)$.
\begin{enumerate}
\item[\eqref{eq:cs1}]
$z'(\theta) y(\theta)=0$: we have $y(\theta)=0$ for $\theta\in[0,\mu)$ and $z'(\theta)=0$ for $\theta\in(\mu,1]$.

\item[\eqref{eq:cs2}]
$\left[(1-\theta) z(\theta) + 2\int_{0}^{\theta} z(\lambda) d \lambda - 1 \right] w(\theta) = 0$: we have $(1-\theta) z(\theta) + 2\int_{0}^{\theta} z(\lambda) d \lambda - 1 = (1-\theta)^2 + 2 (\theta-\frac{\theta^2}{2}) - 1 = 0$ for $\theta\in[0,\mu)$ and $w(\theta)=0$ for $\theta\in(\mu,1]$.

\item[\eqref{eq:cs3}]
$\left[ z(1) + \frac{3}{2} \int_{0}^{1} z(\theta) d \theta - 1 \right] \gamma = 0$: we have $z(1) + \frac{3}{2} \int_{0}^{1} z(\theta) d \theta - 1 = (1-\mu) + \frac{3}{2} \left(1-\mu+\frac{\mu^2}{2} \right) - 1 = 0$ by the definition of $\mu$.

\item[\eqref{eq:cs4}]
$\left[ (1-\theta) w(\theta) + 2 \int_{\theta}^{1} w(\lambda) d \lambda + \frac{3\gamma}{2} + y'(\theta) - 1\right] z(\theta) = 0$: for $\theta\in[0,\mu)$, we have
\begin{align*}
\txts (1-\theta) w(\theta) + 2 \int_{\theta}^{1} w(\lambda) d \lambda + \frac{3\gamma}{2} + y'(\theta) - 1
= \frac{2(1-\mu)^2}{(5-3\mu)(1-\theta)^{2}} + 2 \int_{\theta}^{\mu} w(\lambda) d \lambda + \frac{3(1-\mu)}{5-3\mu} + 0 - 1 = 0,
\end{align*}
and for $\theta\in(\mu,1]$, we have
\begin{align*}
\txts (1-\theta) w(\theta) + 2 \int_{\theta}^{1} w(\lambda) d \lambda + \frac{3\gamma}{2} + y'(\theta) - 1
= \frac{3\gamma}{2} + y'(\theta) - 1
= \frac{3(1-\mu)}{5-3\mu} + \frac{2}{5-3\mu} - 1 = 0.
\end{align*}

\item[\eqref{eq:cs5}]
$(\gamma - y(1)) z(1)=0$: we have $\gamma-y(1) = \frac{2(1-\mu)}{5-3\mu} - \frac{2(1-\mu)}{5-3\mu} = 0$.
\end{enumerate}
Moreover, the optimal value of $\LPi$ is $\int_{0}^{1} z(\theta) d\theta =1 - \mu + \frac{\mu^2}{2} = \frac{2(5-\sqrt{7})}{9}\approx 0.523$.
\end{proof}

\begin{proofof}{Theorem~\ref{th:ratio}}
The expected ratio of \ranking is lower bounded by the optimal value of $\LPn$. Hence, the theorem follows from Lemmas~\ref{lemma:lp_relax} and~\ref{lemma:opt}.
\end{proofof}



{
\bibliography{ref_s}
\bibliographystyle{plain}
}

\appendix

\section{Issues with the Experimental Results on $LP(n)$} \label{appendix:lp}

We ran experiments on the LP described in Section III.B of~\cite{Goel12} and obtained the following results. The source code (in MathProg format) is available at:

\url{http://i.cs.hku.hk/~algth/project/online_matching/issue.html}.


\begin{table}[h]
\centering
\begin{tabular}{|c|r|}
\hline
n = 20 & 0.5024 \\ \hline
n = 50 & 0.5010 \\ \hline
n = 100& 0.5005 \\ \hline
n = 200& 0.5003 \\ \hline
n = 300& 0.5002 \\ \hline
n = 400& 0.5001 \\ \hline
\end{tabular}
\caption{Our Experimental Results on $LP(n)$ in~\cite{Goel12}}
\end{table}

\vspace{0.5cm}
Hence, it is impossible to use $LP(n)$ to show that the performance ratio is
larger than 0.5002.

\ignore{
We could not reproduce the proof of Lemma 10 in~\cite{Goel12}, which stated that the optimal value of $LP(n)$ is monotonically increasing and was not supported by the above results. The proof idea is to construct a solution for $LP(k)$ from a solution for $LP(n)$, when k divides n. In their proof (in the full version of their paper), they divided the $x_s$'s into blocks of size $\frac{n}{k}$ and use the average as the new values. We think the argument to derive the inequalities generated by Lemma 7 in~\cite{Goel12} is problematic. They changed one of the coefficients from $n/t$ to $1/n$ without explanation. After the construction, the inequalities might be violated.
}

\section{Hardness Result} \label{sec:hardness}


\begin{figure}[h]
\centering
\includegraphics[width=200pt]{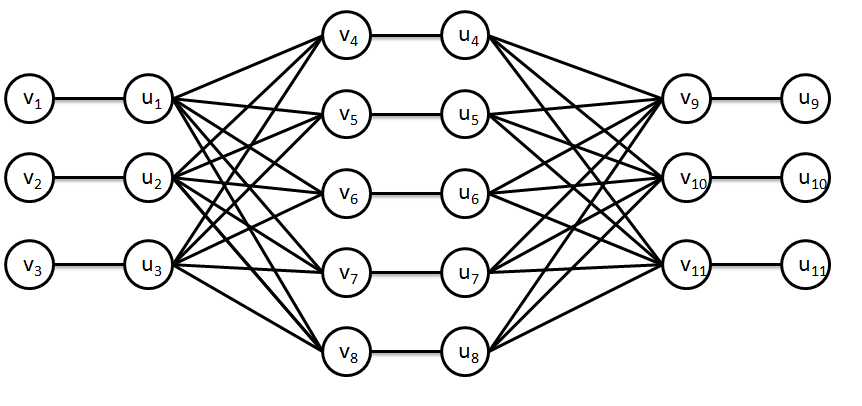}
\caption{Double Bomb Graph}
\label{picbomb}
\vspace{20pt}
\end{figure}
In this section, we show that we can slightly improve the hardness result in~\cite{Karande2011} by adjusting the parameter. An example of the graph is shown in \ref{picbomb}. We define the graph as follows:

Let $G$ be a bipartite graph over $2(3+\epsilon)n$ vertices ($u_i$'s and $v_i$'s). Define the edges by adjacency matrix $A$. ($A[i][j]=1$ if there is an edge between $u_i$ and $v_j$.)

\[
A[i][j] = \begin{cases}
1 & \text{if }i=j\\
1 & \text{if }i\in [1, n]\text{ and }j\in (n, (2+\epsilon)n]\\
1 & \text{if }i\in (n, (2+\epsilon)n]\text{ and }j\in ((2+\epsilon)n, (3+\epsilon)n]\\
0 & \text{otherwise}
\end{cases}
\]

We run experiments on different $n$'s and $\epsilon$'s and get the following result.

\begin{table}[h]
\centering
\begin{tabular}{|c|c|c|c|c|c|}

\hline
				& $n=20$ & $n=50$ & $n=100$ & $n=200$ & $n=500$\\ \hline
$\epsilon=0.63$& $0.7314$ & $0.7267$ & $0.7253$ & $0.7244$ & $0.7240$\\ \hline
\end{tabular}
\end{table}

We observe that when $\epsilon=1-1/e$ the ratio is minimized for this kind of graph. It is close to $0.724$ in this case. We leave it as future work to analyze it theoretically.

\end{document}